\documentclass{article}
\usepackage{amsfonts}

\usepackage{graphicx}
\usepackage{amsmath}
\usepackage{mathrsfs}

\newtheorem{theorem}{Theorem}

\newtheorem{corollary}[theorem]{Corollary}

\newtheorem{definition}[theorem]{Definition}

\newtheorem{proposition}[theorem]{Proposition}
\newtheorem{remark}[theorem]{Remark}

\newenvironment{proof}[1][Proof]{\textbf{#1.} }{\ \rule{0.5em}{0.5em}}

\begin{document}

\title{Quantum Anomalies and Logarithmic Derivatives}
\author{J.E. Gough$^{(a)}$, T. S. Ratiu$^{(b)}$, and O. G. Smolyanov$^{(c)}$\\
\quad \\
\noindent a) Institute of Mathematics, Physics, and Computer Sciences,\\
Aberystwyth University, Great Britain\\
email: \texttt{jug@aber.ac.uk}\\
b) Section de Mathématiques and Bernoulli Center,\\
 \'{E}cole
Polytechnique F\'{e}d\'{e}rale de Lausanne,\\
 CH-1015 Switzerland\\
email: \texttt{tudor.ratiu@epfl.ch}\\
\indent c) Mechanics and Mathematics Faculty,\\
Moscow State University, Moscow, 119991 Russia\\
email: \texttt{smolyanov@yandex.ru}}
\date{}
\maketitle

\begin{abstract}
This papers deals with connections between quantum anomalies and
transformations of Feynman pseudo-measures. Mathematical objects
related to the notion of the volume element in an
infinite-dimensional space considered in the physics literature
[1] are considered  and disagreement in the related literature
regarding the origin of quantum anomalies is explained.
\end{abstract}

\section{Introduction}

A quantum anomaly is the violation of a symmetry (see [4]) with respect to
some group of transformations under quantization. That is, a situations where
a classical Hamiltonian system exhibits an invariance with respect to some transformations
however the same is not true for its quantization (see [2]).

There is some degree of disagreement in the related literature regarding the origin
of quantum anomalies. An explanation is given by in the 2004 edition of Fujikawa and Suzuki's
well-known book [1], however,
in the 2006 monograph of Cartier and DeWitt-Morette [2, page 352], it is claimed that this
description as to why quantum anomalies occur given is
incorrect. The second edition of [1] in 2013 however did not address this criticism of
their book. Our contribution in this paper is to analyze the problem from the point of
view of Feynman pseudo-measures (defined in the next section), and conclude
that description of the emergence of quantum anomalies given in [1] is essential correct.

We shall use the fact that the transformations of a functional Feynman integral
(i.e., an integral with respect to a Feynman pseudo-measure)
are determined by transformations of two distinct objects:

\begin{enumerate}
    \item
 The first of these
is the product of a Feynman pseudo-measure with a certain function integrable
with respect to this pseudo-measure, which is, in turn, the product of the
exponential of a part of the classical action and the initial condition. The
exponential of the other part of the action determines the Feynman
pseudo-measure. If the action and the initial condition are invariant with
respect to some transformation, then this object is invariant with respect
to this transformation as well.
\item
The second object is a determinant, which plays the role of a Jacobian; this
determinant may differ from unity even in the case where the action and the
initial conditions are invariant with respect to phase transformations; of
course, in this case, the Feynman integral is non-invariant as well.
\end{enumerate}

What is said above agrees in spirit with the viewpoint of [1]. The
authors in [2] proposed however to compensate for this determinant
by multiplying the measure with respect to which the integration
is performed\footnote{Of course this counterpart of the classical
Lebesgue measure  does not exist in the infinite-dimensional case
according to the well-known theorem of Weil, but it is not
important in this context because the counterpart can be
considered as a pseudomeasure.} by an additional factor, which, of
course, is equivalent to multiplying the integrand by the same
factor.

In this paper, we consider families of
transformations of the domain of a (pseudo)measure depending on a real
parameter and show that such a compensation is impossible; for this purpose,
we use differentiation with respect to this parameter. The paper is
organized as follows. First, we recall two basic definitions of the
differentiability of a measure and, more generally, a pseudo-measure
(distribution); then, we give explicit expressions for the logarithmic
derivatives of measures and pseudo-measures with respect to transformations
of the space on which they are defined. The application of these expressions%
\footnote{%
Among other things, these formulas lead to infinite-dimensional versions of
both the First and Second Noether Theorems} makes it possible to obtain a mathematically
correct version of results of [1] concerning quantum anomalies. After this,
we discuss the approach to explaining the same anomalies proposed in [2]. We
also discuss mathematical objects related to the notion of the volume
element in an infinite-dimensional space considered in physics literature
(including [2]). We concentrate on the algebraic structure of problems,
leaving aside most assumptions of analytical character.

\section{Differentiation of Measures and Distributions}

In this section we recall some definitions, conventions, and
results about differentiable measures and distributions on
infinite-dimensional spaces in a form convenient for our purposes.

\begin{definition}
Given a locally convex space (LCS) $E$, we denote the $\sigma $-algebra of
Borel subsets of $E$ by $\mathcal{B}_{E}$, and the vector space of countably
additive (complex) measures on $E$ by $\mathfrak{M}_{E}$ . We say that the
vector space $C$ of bounded Borel functions on $E$ determines a norm if, for
any measure in $\mathfrak{M}_{E}$, its total variation $\mu \in \mathfrak{M}%
_{E}$ satisfies the condition $\Vert \mu \Vert _{1}=\sup \bigg\{\int ud\mu
\,:\,u\in C,\,\Vert u\Vert _{\infty }\leq 1\bigg\}$, where $\Vert u\Vert
_{\infty }=\sup \{|u(x)|:x\in E\}$.
\end{definition}

A \textit{Hilbert subspace} of an LCS $E$ is defined as a vector
subspace $H$ of $E$ endowed with the structure of a Hilbert space
such that the topology induced on $H$ by the topology of $E$ is
weaker than the topology generated by the Hilbert norm. We now
define the notion of smooth maps, of  a LCS, along a Hilbert
subspace.

\begin{definition}
A mapping $F$ of an LCS $E$ to an LCS $G$ is said to be smooth
along a Hilbert subspace $H$ of $E$ (or $H$-smooth) if it is
infinitely differentiable along $H$ and both the mapping $F$ and
all of its derivatives (along $H$) are continuous on $E$, provided
that the spaces in which the derivatives take values are endowed
with the topologies of uniform convergence on compact subsets of
$H$.
\end{definition}

The class of vector fields on a LCS is introduced next.

\begin{definition}
A vector field on an LCS $E$ is a mapping $h:E\mapsto E$; we denote the set
of vector fields on $E$ by $Vect(E)$. The derivative along a vector field $%
h\in Vect(E)$ of a function $f$ defined on $E$ is the function on $E$
denoted by $f^{\prime }h$ and defined by
\begin{equation*}
(f^{\prime }h)(x):=f^{\prime }(x)\,h(x),\qquad \forall x\in E,
\end{equation*}
where $f^{\prime }(x)$ is the G\^{a}teaux derivative of $f$ at the point $x$.
\end{definition}

Let $\epsilon >0$, and let $S$ be a mapping of the interval $(-\epsilon
,\epsilon )$ to the set of $\mathcal{B}_{E}$-measurable self-mappings of $E$
for which $S(0)=id$; suppose that $\tau $ is a topology on $E$ compatible
with the vector space structure. A measure $\nu \in E$ is said to be $\tau $%
-differentiable along $S$ if the function
\begin{eqnarray}
f &:&(-\epsilon ,\epsilon )\mapsto (\mathfrak{M}_{E},\tau ),  \notag \\
&:&t\mapsto S(t)_{\ast }\nu :=\nu (S(t)^{-1}),  \label{eq:f}
\end{eqnarray}
is differentiable at $t=0$ (the symbol $S(t)_{\ast }\nu $ denotes the image
of $\nu $ under the mapping $S(t)$); in this case, we denote $f^{\prime }(0)$
by $v_{S}^{\prime }$ and call it the \textit{derivative of the measure} $\nu
$ along $S$. If, in addition, $f^{\prime }\ll f(0)$ then we may furthermore
define the logarithmic derivatives for the density.

\begin{definition}
Suppose that the mapping (\ref{eq:f}) is differentiable at $t=0$, and that
the measure $f^{\prime }(0)$ is absolutely continuous with respect to $f(0)$%
, then its density with respect to the measure $f(0)$ is called the $\tau $%
-logarithmic derivative of the measure $\nu $ along $S$ and denoted by $%
\beta _{S}^{v}$.
\end{definition}

If $k\in E$ and $S(t)(x):=x-tk$, then a measure $\nu $ which is $\tau $%
-differentiable along $S$ is said to be $\tau $-differentiable along $k$,
and $\nu ^{\prime }k$ is defined by $\nu ^{\prime }k=v_{S}^{\prime }$; the $%
\tau $-logarithmic derivative of the measure $\nu $ along $S$ is called the $%
\tau $-logarithmic derivative of $\nu $ along $k$ and denoted by $\beta
^{\nu }(k,\cdot )$. The $\tau $-differentiability of a measure along a
vector field $h$ and its $\tau $-logarithmic derivative along $h$ (denoted
by $\beta _{h}^{\nu }$) are defined in a similar way: we set
\begin{equation*}
S(t)(x):=x-t\, h(x).
\end{equation*}

If a measure $\nu $ is $\tau $-differentiable along each $k\in E$, then it
can be shown that the mapping $E\ni k\mapsto \nu ^{\prime }k$ is linear; the
corresponding vector-valued measure
\begin{equation*}
\nu ^{\prime }:\mathcal{B}_{E}\ni A\mapsto [ k\mapsto (\nu ^{\prime
}k)(A)]
\end{equation*}
is called the $\tau $-derivative of $\nu $ over the subspace $H$.
If, for any $k\in H$, there exists a $\tau $-logarithmic
derivative measure $\nu $ along $k$, then the mapping $H\ni
k\mapsto \beta _{\nu }(k,\cdot )$ is linear; it is called the
$\tau $-\textit{logarithmic derivative of} $\nu $ over (or along)
$H$ and denoted by $\beta _{\nu }$.

\begin{remark}
If the measure $\nu $ has a logarithmic derivative over a subspace $E$ and $%
h(x)\in H$ for all $x\in E$, then, contrary to what one might expect,
\begin{equation*}
\beta _{S}^{\nu }(x)\neq \beta _{\nu }(h(x),x)
\end{equation*}
in the general case (see below).
\end{remark}

\begin{remark}
If $\tau $ is the topology of convergence on all sets, then any measure that is
$\tau $-differentiable along $S$ will have a logarithmic derivative along $\tau $
(see [12]), however this may not ne the case for weaker topologies. An example is
as follows. In the case where the LCS $E$ is also a Radon space\footnote{%
A topological space $E$ is called a Radon space if any countably additive
Borel measure $\nu $ on $E$ is Radon; this means that, for any Borel subset $%
A$ of $E$ and any $\epsilon >0$, there exists a compact set $K\subset A$
such that $\nu (A\ K)<\epsilon $. If $E$ is a completely regular Radon
space, then the space of all bounded continuous functions on $E$ is in
natural duality with $\mathfrak{M}_{E}$.}, let $S$ be the space of bounded
continuous functions on $E$, and let $\tau _{C}$ be the weak topology on $%
\mathfrak{M}_{E}$ determined by the duality between $C$ and $\mathfrak{M}_{E}
$. Then a measure $\tau _{C}$-differentiable along $S$ may have no
logarithmic derivative along $S$ (even in the case $E=\mathbb{R}^{1}$).
\end{remark}

\begin{definition}
Let $C$ be a norm-defining vector space of $H$-smooth functions on $E$
bounded together with all derivatives. A measure $\nu $ is said to be $C$%
-differentiable along a vector field $h\in Vect(E)$ if there exists a
measure $ \nu_{h}^{\prime }$ such that, for any $\varphi \in C$, we have
the integration by parts formula
\begin{equation*}
\int \varphi ^{\prime }(x)(x)\nu (dx)=-\int \varphi (x)\left( \nu
_{h}^{\prime }\right) (dx).
\end{equation*}
The Radon-Nikodym density of $\nu _{h}^{\prime }$ with respect to $\nu $ (if
it exists) is called the $C$-logarithmic derivative of the measure $\nu $
along $h$; if $h(x)=h_{0}\in E$ for all $x\in E$, then, as above, the $C$%
-logarithmic derivative of $\nu $ along $h$ is called the $C$-logarithmic
derivative of $\nu $ along $h_{0}$
\end{definition}

Note that we denote $C$-logarithmic derivatives by the same symbols as $\tau
$-logarithmic derivatives introduced above as there should be no confusion.

Suppose that a vector field $h_S $ is determined by $h_S(x): = S^{\prime}(0)x
$. Then the following proposition is valid.

\begin{proposition}
A measure $\nu$ is $\tau_C$-differentiable along $S$ if and only if it is $C$%
-differentiable along $h_ S$. In this case, $\beta_{h_S}^\nu = \beta_S^\nu$,
where $\beta_{h_S}^\nu$ is the $C$-logarithmic derivative of $\nu$ along $h_S
$ and $\beta_S^\nu$ is the $\tau_C$-logarithmic derivative of $\nu$ along $S$%
.
\end{proposition}

\begin{proof}
This follows from the change of variable formula. Suppose that $\varphi \in C
$ and, as above, $f(t) : = (S(t))_* \nu$. Then
\begin{eqnarray*}
&& \lim_{t\to 0} t^{-1}\int_E \varphi(x)(f(t))(dx) -
\int_E\varphi(x)(f(0))(dx) \\
&=& \lim_{t\to 0} t^{-1}\int_E \varphi(x) ((S(t))_* \nu)(dx) -
\int_E\varphi(x)\nu(dx) \\
&=& \lim_{t\to 0}t^{-1}\int_E \left(\varphi(S(t)) - \varphi(x)\right)
\nu(dx) = \int_E\varphi^{\prime}(x) _S(x) \nu(dx),
\end{eqnarray*}
which gives the stated result.
\end{proof}

\begin{corollary}
\label{cor:No1} Let $S_{1}$ be another mapping of the interval $(-\epsilon
,\epsilon )$ to $\mathcal{B}_{E}$ with the same properties as $S$. If $%
h_{S}=h_{S_{1}}$, then the measure $\nu $ is $\tau _{C}$-differentiable
along $S$ if and only if it is $\tau _{C}$-differentiable along $S_{1}$.
\end{corollary}

\begin{remark}
It is natural to say that the measure $\nu$ is invariant with respect to $S$
if $\beta_S^\nu=0$.
\end{remark}

\begin{theorem}
\label{thm:No1} Suppose that a measure $\nu \in \mathfrak{M}_{E}$ has a $%
\tau _{C}$-logarithmic derivative over a subspace $H$, and let $h$ be a
vector field on $E$ taking values in $H$ . Then
\begin{equation*}
\beta _{h}^{\nu }(x)=\beta ^{\nu }\left( h(x),x\right) +\mathrm{tr}\, h^{\prime
}(x),
\end{equation*}
where $h^{\prime }$ is the derivative of the mapping $h$ over the subspace $H
$.
\end{theorem}

\begin{proof}
Suppose that $\varphi \in C$ and $h$ is a vector field on $E$, and let $\mu $
be the $E$-valued measure defined by
\begin{equation*}
\mu ^{\prime }=h^{\prime }(\cdot )\varphi (\cdot )\nu +\varphi (\cdot )\nu
^{\prime }\otimes h(\cdot )+\varphi ^{\prime }(\cdot )\otimes h(\cdot )\nu .
\end{equation*}
Applying the Leibniz rule to the derivative of $\mu $ over the subspace $H$,
we obtain
\begin{equation*}
\mathrm{tr}\mu ^{\prime }=\left( \varphi (\cdot )\mathrm{tr}h^{\prime }(\cdot
)\right) \nu +\varphi (\cdot )\beta ^{\nu }(h(\cdot ),\cdot )\nu +\varphi
^{\prime }(\cdot )h(\cdot )\nu .
\end{equation*}
Each summand in this relation is a measure whose values are operators on $H$%
; calculating the traces of these operators, we obtain
\begin{equation*}
\mathrm{tr}\mu ^{\prime }=\left( \varphi (\cdot )\mathrm{tr}h^{\prime }(\cdot
)\right) \nu +\varphi (\cdot )\beta ^{\nu }(h(\cdot ),\cdot )\nu +\varphi
^{\prime }(\cdot )h(\cdot )\nu .
\end{equation*}
Since $\int_{E}\mu ^{\prime }(dx)=0$ and, therefore, $\int_{E}\mathrm{tr}\mu
^{\prime }(dx)=0$, it follows that
\begin{equation*}
\int_{E}\varphi ^{\prime }(x)h(x)\nu (dx)=-\int_{E}\varphi (x)\big[
\beta ^{\nu }(h(x),x)+\mathrm{tr}h^{\prime }(x)\big]\nu (dx).
\end{equation*}

This means that the required relation holds.
\end{proof}

Both the definitions given above and the algebraic parts of proofs can be
extended to distributions (in the Sobolev-Schwartz sense) defined as
continuous linear functionals on appropriate spaces of test functions. The
difference is that the integrals of functions with respect to measures
should be replaced by values of these linear functionals at functions, and
instead of the change of variables formula for integrals, the definition of
the transformation of a distribution generated by a transformation of the
space on which the test functions are defined should be used.

\section{Quantum Anomalies}

In fact, quantum anomalies arise because the second term in the
relation of Theorem \ref{thm:No1} proved above is the same for all
Feynman (pseudo)measures. Indeed, by virtue of Leibniz rule, the
logarithmic derivative (both over a subspace and along a vector
field) of the product of a function and a measure is the sum of
the logarithmic derivatives of the factors; therefore, a measure
$\nu $ whose logarithmic derivative along a vector field is given
by the expression in Theorem \ref{thm:No1} can formally be taken
for the product of a function $\psi _{n}u$ whose logarithmic
derivative over the subspace $H$ coincides with the logarithmic
derivative of the measure $\nu $ over this subspace and a measure
$\eta $ whose logarithmic derivative over the same subspace
vanishes. If $E$ is finite-dimensional and $H$ coincides with $E$,
then such a function and a measure indeed exist; moreover, $\eta $
turns out to be the Lebesgue measure, and $\psi _{\nu }$ is the
density of $\nu $ with respect to it.

But in the infinite-dimensional case, there exist no exact counterpart of
the Lebesgue measure; nevertheless, an analogue of density, called the
generalized density of a measure, does exist [3, 10, 12], although its
properties are far from those of usual density, and the corresponding
distribution can be regarded as an analogue of the Lebesgue measure. It is
this distribution that should be considered as a formalization of the term -
volume element - used in [2, p. 362]. We however emphasize that the contents
of this paper depends on the properties of neither this distribution nor the
generalized density.

Let $Q$ be a finite-dimensional vector space being the configuration space
of a Lagrangian system with Lagrange function $L:Q\times Q\mapsto \mathbb{R}$
defined by $L(q_{1},q_{2}):=\eta (q_{1},q_{2})+b(q_{2})$, where $b$ is a
quadratic functional (the kinetic energy of the system). We assume that the
Lagrange function $L$ is nondegenerate (hyperregular), i.e., the
corresponding Legendre transform is a diffeomorphism, so that it
determines a Hamiltonian system with Hamiltonian function $\mathcal{H}%
:Q\times P\mapsto \mathbb{R}$, where $P=Q^{\ast }$.

For $t>0$, by $E_{t}$ we denote the set of continuous functions on $[0,t]$
taking values in $Q$ and vanishing at zero and by $H_{t}$, the Hilbert
subspace of $E_{t}$ consisting of absolutely continuous functions on $[0,t]$
with square integrable derivative; the Hilbert norm on $E_{t}$ is defined by
\begin{equation*}
\Vert f\Vert _{H_{t}}^{2}:=\int_{0}^{t}\Vert f^{\prime }(\tau )\Vert
_{Q}^{2}d\tau
\end{equation*}
where $f\in E_{t}$ and $\Vert \cdot \Vert _{Q}$ is the Euclidean norm on $Q$%
. Finally, by $\mathscr{S}(t)(f)$ we denote the classical action defined as
the functional on $H_{t}$ determined by the Lagrangian function $L$
according to
\begin{equation*}
\mathscr{S}(t)(f):=\int_{0}^{t}L(f(\tau ),\dot{f}(\tau ))d\tau .
\end{equation*}

The Schr\"{o}dinger quantization of the Hamiltonian system generated by the
Lagrangian system described above yields the Schr\"{o}dinger equation $%
\mathrm{i}\dot{\psi}(t)=\widehat{\mathcal{H}}\psi (t)$, where $\widehat{%
\mathcal{H}}$\ is a self-adjoint extension of a pseudo-differential operator
on $\mathcal{L}_{2}(Q)$ with symbol equal to the Hamiltonian function $H$
generated by the Lagrange function $L$ . The solution of the Cauchy problem
for this equation with initial condition $f_{0}$ is
\begin{equation}
\psi (t)(q)=\int_{E_{t}}e^{\mathrm{i}\int_{0}^{t}\eta (\psi (\tau )+q,\dot{%
\psi}(\tau ))d\tau }f_{0}(\psi (t)+q)\phi _{t}(d\psi ),
\label{feynman_integral}
\end{equation}
where $\phi _{t}$ is the Feynman pseudo-measure on $E_{t}$ (the exponential
under the integral sign is well defined on the space $H_{t}$).

Let $W_{t}$ be the pseudo-measure on $E_{t}$ defined as the product of the
exponential in the above integral and the psuedo-measure $\phi _{t}$. The
following theorem is valid.

\begin{theorem}
\label{thm:No2} The logarithmic derivative of the pseudo-measure $W_{t}$
along $H_{t}$ exists and is determined by
\begin{equation*}
\beta ^{W_{t}}(k,\psi )=\mathrm{i}\int_{0}^{t}\left[L_1^{\prime}
(\psi
(\tau )+q,\dot{\psi}(\tau )) \,k(\tau )+  L_2^{\prime}(\psi (\tau )+q,\dot{\psi%
}(\tau )) \, \dot{k}(\tau )\right] d\tau
\end{equation*}
where  $k\in H_{t}$, $\psi \in E_{t}$.
\end{theorem}

\begin{corollary}
If $h$ is a vector field on $E_ t$ taking values in $H_ t $, then the
logarithmic derivative of the pseudo-measure $W_ t$ along $h$ is determined
by
\begin{align*}
\beta_h^{W_t}(\psi) = &\mathrm{i}\int_0^t \left[L_1^{\prime}\left(h(\psi)(%
\tau) + q, \frac{d h(\psi)}{d\tau}(\tau) \right)h(\psi)(\tau) \right. \\
& \quad + \left. L_2^{\prime}\left(h(\psi)(\tau) + q, \frac{d h(\psi)}{d\tau}%
(\tau)\right) \frac{d h(\psi)}{d\tau}(\tau) \right]d\tau + \mathrm{tr} \,
h^{\prime}(\psi).
\end{align*}
\end{corollary}

A similar assertion is valid for the logarithmic derivative along a family $%
S(\alpha )$ of transformations of the space $E_{t}$ depending on a parameter
$\alpha \in (-\epsilon ,\epsilon )$.

It follows from Corollary \ref{cor:No1} that if the classical action $%
\mathscr{S}(t)$ is invariant with respect to a family $S ( \alpha )$, $%
\alpha \in (-\epsilon, \epsilon)$, of transformations of the space
$E_t$, then the logarithmic derivative $\beta_S^W(\psi)$ does not
necessarily vanish. In turn, this means that if a Lagrange
function is invariant with respect to a family of transformations
of the configuration space and hence the action is invariant with
respect to the corresponding gauge transformations then  the
solution of the corresponding Schr\"{o}dinger equation is not
necessarily invariant with respect to the same gauge
transformations.

\begin{remark}
For each family $S(\alpha )$ of transformations of the space $E_{t}$, we can
obtain an explicit expression for the transformations of the psuedo-measure $%
W_{t}$ generated by the transformations $S(\alpha )$ by solving the equation
$\dot{g}(\alpha )=\beta _{S(\alpha )}^{\nu }g(\alpha )$ (see [10]). It
follows [10] that if $\mathrm{tr}\left( h_{S(\alpha )}\right) ^{\prime
}(\psi )>0$, for $\alpha \in \lbrack 0,\alpha _{0}]$, , then $det(S(\alpha
))^{\prime }\neq 1$; this fact does not depend on the classical action.
\end{remark}

\begin{remark}
Using the notion of the generalized density of a pseudo-measure (cf. [3, 10,
11], where only generalized densities of usual measures were considered), we
can say that the pseudo-measure $W_{t}$ is determined by its generalized
density being the exponential in the Feynman integral (\ref{feynman_integral}%
). Moreover, as mentioned above, the expression for the transformations of
the pseudo-measure contain determinants, and the expressions for the
corresponding logarithmic derivatives contain traces, which do not depend on
the generalized densities. This can be interpreted by treating the Feynman
pseudo-measure as the product of its generalized density and distribution
whose transformations are described by the corresponding determinants and
traces, and the logarithmic derivatives of this distribution along constant
vectors vanish. In turn, this allows us to say that the distribution
mentioned above corresponds to the volume element considered in [2].
\end{remark}

\begin{remark}
Thus, the determinants and traces mentioned above cannot be eliminated by
any choice of the integrand and the Feynman pseudo-measure if the
corresponding Feynman integrals are required to represent the solutions of
the corresponding Schr\"{o}dinger equation. Clearly, this contradicts [2, p.
362].
\end{remark}

\begin{remark}
If $E$ is a  superspace, then, instead of traces and determinants,
we should use supertraces and superdeterminants.
\end{remark}

\section*{Acknowledgements}

\textit{T.S. Ratiu acknowledges the partial support of Swiss National Science
Foundation, grant no. Swiss NSF 200021-140238. J. Gough and O.G. Smolyanov
acknowledge the partial support of the Swiss Federal Institute of Technology
in Lausanne (EPFL). O.G. Smolyanov acknowledges the support of the Russian
Foundation for Basic Research, project no. 14-01-00516.}

\end{document}